\documentclass[a4paper,UKenglish,cleveref, autoref, thm-restate]{lipics-v2021}

\title{Graph Threading} 

\author{Erik D. {Demaine}}{Computer Science and Artificial Intelligence Lab, Massachusetts Institute of Technology, USA \and \url{https://erikdemaine.org/}}{edemaine@mit.edu}{https://orcid.org/0000-0003-3803-5703}{}

\author{Yael {Kirkpatrick}}{Department of Mathematics, Massachusetts Institute of Technology, USA \and \url{https://yaelkirk.github.io/}}{yaelkirk@mit.edu}{https://orcid.org/0009-0007-6718-7390}{NSF Graduate Research Fellowship under Grant No. 2141064.}

\author{Rebecca {Lin}}{Computer Science and Artificial Intelligence Lab, Massachusetts Institute of Technology, USA \and \url{https://rebeccayelin.github.io/}}{ryelin@mit.edu}{https://orcid.org/0000-0003-4747-4978}{MIT Stata Family Presidential Fellowship.}

\authorrunning{Erik D. Demaine, Yael Kirkpatrick, and Rebecca Lin} 

\Copyright{Erik D. Demaine, Yael Kirkpatrick, and Rebecca Lin} 



\category{} 

\acknowledgements{We thank Anders Aamand, Kiril Bangachev, Justin Chen, Alison Martin, Surya Mathialagan, and Zhecheng Wang for insightful discussions. We also thank anonymous reviewers for their helpful comments.}

\EventEditors{Venkatesan Guruswami}
\EventNoEds{1}
\EventLongTitle{15th Innovations in Theoretical Computer Science Conference (ITCS 2024)}
\EventShortTitle{ITCS 2024}
\EventAcronym{ITCS}
\EventYear{2024}
\EventDate{January 30 to February 2, 2024}
\EventLocation{Berkeley, CA, USA}
\EventLogo{}
\SeriesVolume{287}
\ArticleNo{37}

\usepackage{graphicx,psfrag}

\newcommand{\sect}[1]{Section~\ref{#1}}
\newcommand{\eqn}[1]{Equation~\ref{#1}}

\newcommand{\fig}[1]{Figure~\ref{#1}}
\newcommand{\cstr}[1]{(\ref{#1})}
\newcommand{\cfourstar}{(\ref{c4s})}

\newcommand{\threading}{\textsc{Optimal Threading}}
\newcommand{\ip}{\textsc{Optimal Local Threading}}

\newcommand{\double}{\textsc{Double Threading}}

\newcommand{\defn}[1]{\textbf{\textit{\boldmath #1}}}

\newcommand{\set}[1]{\left\{ #1 \right\}}

\renewcommand{\subset}{\subseteq}

\usepackage{algorithm,algpseudocode}
\usepackage{mathtools}
\usepackage{stackrel}
\usepackage{enumitem}

\hideLIPIcs
\nolinenumbers

\begin{document}

\maketitle

\begin{abstract}
Inspired by artistic practices such as beadwork and himmeli, we study the problem of \emph{threading} a single string through a set of tubes, so that pulling the string forms a desired graph. More precisely, given a connected graph (where edges represent tubes and vertices represent junctions where they meet), we give a polynomial-time algorithm to find a minimum-length closed walk (representing a threading of string) that induces a connected graph of string at every junction. The algorithm is based on a surprising reduction to minimum-weight perfect matching. Along the way, we give tight worst-case bounds on the length of the optimal threading and on the maximum number of times this threading can visit a single edge. We also give more efficient solutions to two special cases: cubic graphs and the case when each edge can be visited at most twice.
\end{abstract}

\section{Introduction}
\label{sec:introduction}

Various forms of art and craft combine tubes together by threading cord through them to create a myriad of shapes, patterns, and intricate geometric structures.
In beadwork~\cite{Green_2018}, artists string together beads with thread or wire. In traditional ``straw mobile'' crafts~\cite{straw-wiki} --- from the Finnish and Swedish holiday traditions of himmeli \cite{Finnish_2013,Jackson_2021} to the Polish folk art of paj\c{a}ki \cite{pajaki-nyt} --- mobile decorations are made by binding straws together with string. Artist Alison Martin has shown experiments where bamboo connected by strings automatically forms polyhedral structures by pulling the strings with a weight \cite{Martin_2021}.

\begin{figure*}[ht]
    \centering \includegraphics[width=\textwidth]{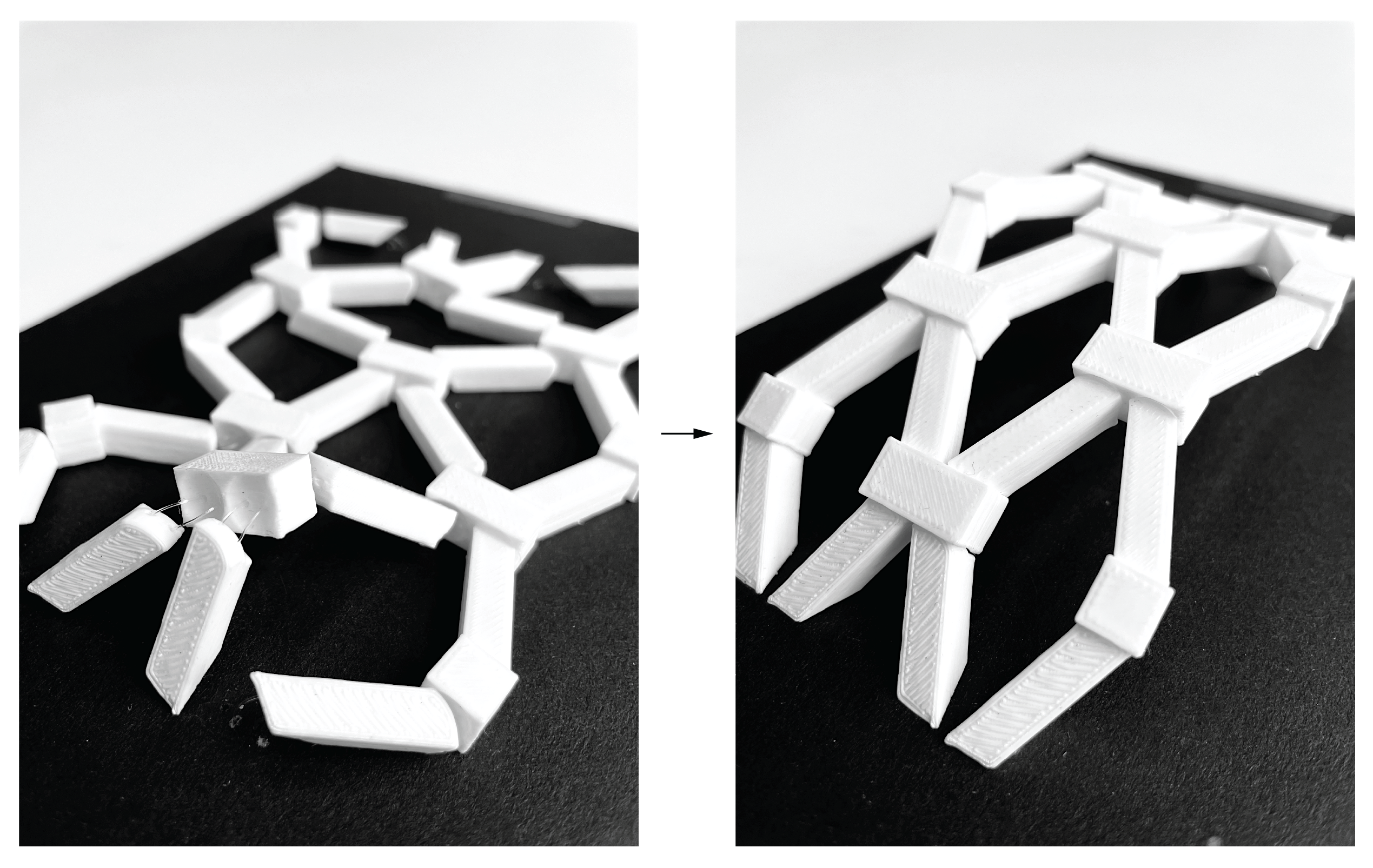}
      \caption{\label{fig:deployable-structure} 
      A deployable structure made from disconnected 3D-printed elements (white) connected by string, which automatically shifts between soft (left) and rigid (right) states by pulling on the endpoints of the string beneath the platform (black). This design was developed by the third author in collaboration with Tomohiro Tachi.}
\end{figure*}

For engineering structures, these techniques offer a promising mechanism for constructing reconfigurable or deployable structures, capable of transforming between distinct geometric configurations: a collection of tubes, loosely woven, can be stored in compact configurations and then swiftly deployed into desired target geometric forms, such as polyhedra, by merely pulling a string taut. Figure~\ref{fig:deployable-structure} shows a prototype of such a structure, illustrating the potential of this approach. The popular ``push puppet'' toy, originally invented by Walther Kourt Walss in Switzerland in 1926 \cite{push-puppets}, also embodies this mechanism.

In contrast to related work~\cite{beady, Jin2019BeadSA}, we study a \textit{theoretical} formulation of these ideas: threading a single string through a collection of tubes to mimic the connectivity of a given graph; refer to Figure~\ref{fig:junction-graphs}.
Consider a connected graph $G=(V, E)$ with minimum vertex degree 2, where each edge $e \in E$ represents a tube and each vertex $v \in V$ represents the junction of tubes incident to $v$.
A \defn{graph threading} $T$ of $G$ is a closed walk through $G$ that visits every edge at least once, induces connected ``junction graphs'', and has no ``U-turns''.
The \defn{junction graph} $J(v)$ of a vertex $v$ induced by a closed walk has a vertex for each tube incident to $v$, and has an edge between two vertices/tubes every time the walk visits $v$ immediately in between traversing those tubes. The threading $T$ must have a connected junction graph $J(v)$ for every vertex $v \in V$, and it must have any \defn{U-turns}: when exiting one tube, the walk next enters a different tube.
Define the \defn{length} $|T|$ of $T$ to be the total length of edges visited by~$T$.
For simplicity, we assume for much of our study that edges (tubes) have unit length --- in which case $|T|$ is the number of edge visits made by~$T$ --- and then generalize to the weighted case with arbitrary edge lengths.

\begin{figure*}[ht]
    \centering \includegraphics[width=\textwidth]{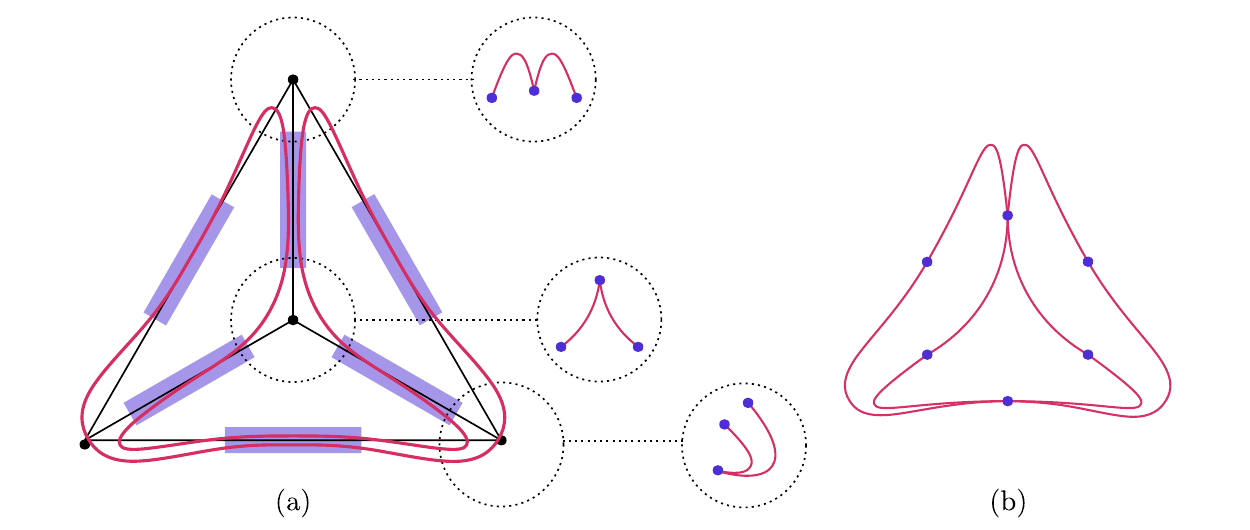}
      \caption{\label{fig:junction-graphs} (a) The closed walk (red) on the graph (black) of a tetrahedron induces junctions graphs (circled on the right) that are connected, and so it is a threading. (b) The union of junction graphs is called the ``threading graph'' (Section~\ref{sec:closed-walk}).}
\end{figure*}

\subparagraph*{Our Results.}

In this paper, we analyze and ultimately solve the \threading\ problem, where the goal is to find a minimum-length threading $T$ of a given graph $G$. Our results are as follows.
\begin{itemize}
    \item In \sect{sec:formulation}, we give a local characterization of threading in terms of local (per-vertex and per-edge) constraints that help us structure our later algorithms and analysis.
    \item In \sect{sec:bounds}, we prove tight worst-case bounds on two measures of an optimal threading~$T$. First, we analyze the minimum length $|T|$ in a graph with unit edge lengths, proving that $2 m - n \leq |T| < 2 m$ where $m$ and $n$ are the numbers of edges and vertices, respectively, and that both of these extremes can be realized asymptotically. Second, we prove that $T$ traverses any one edge at most $\Delta-1$ times, where $\Delta$ denotes the maximum vertex degree in~$G$, and that this upper bound can be realized. The second bound is crucial for developing subsequent algorithms.
    \item In \sect{sec:algorithm}, we develop a polynomial-time algorithm for \threading, even with arbitrary edge lengths, by a reduction to minimum-weight perfect matching.
    \item In \sect{sec:special}, we develop more efficient algorithms for two scenarios: \threading~on cubic graphs, and \double, a constrained version of \threading\ where the threading $T$ is allowed to visit each edge at most twice.
\end{itemize}

\section{Problem Formulation}
\label{sec:formulation}

Let $G = (V, E)$ be a graph with $n = |V|$ vertices and $m=|E|$ edges.
Assume until \sect{sec:weighted} that all edges have unit length.
Recall that a \defn{threading} of $G$ is a closed walk through $G$ that has no U-turns and induces a connected junction graph at each vertex.
As an alternative to this ``global'' definition (a closed walk), we introduce a more ``local'' notion of threading consisting of constraints at each edge and vertex of the graph and prove its equivalence to threading.

We give the intuition of ``local threading'' before we state the formal definition. A local threading assigns a nonnegative integer $x_{uv} \in \mathbb{N}$ for each edge $uv \in E$, which counts the number of times the threading visits or \defn{threads} edge $uv$; we refer to $x_{uv}$ as the \defn{count} of~$uv$.
These integers are subject to four constraints; we argue that these constraints are necessary conditions for a threading: 

\begin{enumerate}
    \item 
    Each edge must be threaded at least once, so $x_{uv}\geq 1$ for all $uv\in E$. 
    
    \item 
    A threading increments the count of \emph{two} edges at junction $v$ every time it traverses $v$, so the sum of counts for all edges incident to $v$ must be even. 
    
    \item 
    Forbidding U-turns implies that, if edge $uv$ is threaded $k$ times, then the sum of counts for the remaining edges incident to $v$ must be at least $k$ to supply these visits. 
    
    \item 
    Because the junction graph $J(v)$ of $v$ is connected, it has at least enough edges for a spanning tree, that is, $d(v)-1$ where $d(v)$ denotes the degree of $v$, so the sum of counts of edges incident to $v$ must be at least $2(d(v) - 1)$. 
\end{enumerate}

More formally:
\begin{definition}[Local Threading]
\label{def:local-threading}
Given a graph $G = (V,E)$, a \defn{local threading} of $G$ consists of integers $\set{x_{uv}}_{uv\in E}$ satisfying the following constraints:

\begin{enumerate}[label=\normalfont{\textbf{(C\arabic*)}},ref=C\arabic*, leftmargin=1cm]
        \item\label{c1}
             $x_{uv} \geq 1$ for all $uv \in E$;
        \item\label{c2}
             $\sum_{u \in N(v)} x_{uv} \equiv 0 \pmod 2$ for all $v \in V$;
        \item\label{c3}
             $\sum_{w\in N(v)\setminus \set{u}} x_{wv}\geq x_{uv}$ for all $uv \in E$; and
        \item\label{c4}
             $\sum_{u \in N(v)} x_{uv} \geq 2(d(v) - 1)$ for all $v \in V$.
    \end{enumerate}

\end{definition}

\ip\, defined as minimizing $\sum_{uv \in E} x_{uv}$, is, in fact, an integer linear program. However, this observation is not helpful algorithmically because integer programming is NP-complete.
Nonetheless, local threading will be a useful perspective for our later algorithms.

The observations above show that any threading $T$ induces a local threading by setting each count $x_{uv}$ to the number of times $T$ visits edge $uv$, with the same length: $|T| = \sum_{uv\in E} x_{uv}$. In the following theorem, we show the converse and, thus, the equivalence of threadings with local threadings:

\begin{theorem}
    We can construct a threading $T$ of $G$ from a local threading $\set{x_{uv}}$ of $G$ such that $T$ visits edge $uv$ exactly $x_{uv}$ times. Hence $|T| = \sum_{uv \in E} x_{uv}$.
\end{theorem}

We shall prove this theorem in two parts. First, we show that it is always possible to form a junction graph at every vertex given a local threading (\sect{sec:junction-graph}). Then we show that a closed walk can be obtained from the resulting collection of junction graphs (\sect{sec:closed-walk}).

\subsection{Constructing a Connected Junction Graph}\label{sec:junction-graph}

Forming a junction graph $J(v)$ at vertex $v$ reduces to constructing a connected graph on vertices $t_1, \dots, t_{d(v)}$, where each vertex represents a tube incident with $v$, with degrees $x_1, \dots, x_{d(v)}$, respectively.
We shall construct $J(v)$ in two steps, first in the case where \cstr{c4} holds with equality (Lemma~\ref{lemma:spanning-tree}) and then in the general case (Lemma~\ref{lemma:junction-graph}).

\begin{lemma}\label{lemma:spanning-tree}
We can construct a tree $S$ consisting of $d$ vertices with respective degrees $x_1, \dots, x_{d} \geq 1$ satisfying $\sum_{i=1}^d x_i = 2(d-1)$ in $O(d)$ time.
\end{lemma}

\begin{proof}
We provide an inductive argument and a recursive algorithm.
In the base case, when $d=2$, $x_1 = x_2 = 1$, and so the solution is a one-edge path.
For $d > 2$, the average $x_i$ value is $\frac{2(d-1)}{d}$ which is strictly between $1$ and~$2$.
Hence there must be one vertex $i$ satisfying $x_i > 1$
and another vertex $j$ satisfying $x_j = 1$.
Now apply induction/recursion to $x'$
where $x'_k = x_k$ for all $k \notin \{i,j\}$,
$x'_i = x_i - 1$, and $x_j$ does not exist (so there are $n-1 < n$ values),
to obtain a tree $S'$. We can construct the desired tree $S$ from $S'$
by adding the vertex $j$ and edge $ij$.

The recursive algorithm can be implemented in $O(d)$ time as follows.
We maintain two stacks: the first for vertices of degree $> 1$ and the second for vertices of degree~$1$.
In each step, we pop vertex $i$ from the first stack, pop vertex $j$ from the second stack, and connect vertices $i$ and $j$. We then decrease $x_i$ by $1$ and push it back onto one of the stacks depending on its new value. This process continues until the stacks are empty. Each step requires constant time, and we perform at most $\sum_{i=1}^d x_i = O(d)$ steps, so the total running time is $O(d)$.
\end{proof}

\begin{lemma}
\label{lemma:junction-graph}
    Given a local threading $\set{x_{e}}$ and a vertex $v \in V$, we can construct a connected junction graph $J(v)$ with no self-loops in $O\big(\sum_{u \in N(v)}x_{uv}\big)$ time.
\end{lemma}

\begin{algorithm}
\caption{Constructing a Connected Junction Graph $J(v)$}\label{alg:junction-graph}
\begin{enumerate}
    \item \label{step:1} $R \leftarrow \emptyset$ \Comment{set of ``redundant'' edges}
    \item \label{step:2} $x_{i}' \leftarrow x_{i}~\text{for all}~i \in \set{1, \dots, d(v)}$
    \item \label{step:3} Repeat until $\sum_{i=1}^{d(v)} x_i' = 2(d(v)-1)$:
    \begin{enumerate}
        \item \label{step:3a}
        $x_{\alpha}' \leftarrow x_{\alpha}' - 1$ where
        $x_{\alpha}' = \max_{i=1}^{d(v)} x_i'$, breaking ties arbitrarily
        \item \label{step:3b}
        $x_{\beta}' \leftarrow x_{\beta}' - 1$ where
        $x_{\beta}' = \max_{i \in \set{1, \dots, d(v)} \setminus \set{\alpha}} x_i'$, breaking ties arbitrarily
        \item $ R \leftarrow R \cup \set{t_\alpha t_\beta}$
    \end{enumerate}
    \item \label{step:4} Compute tree $S$ on vertices $t_1, \dots, t_{d(v)}$ with degrees $x'_1, \dots, x'_{d(v)}$ (Lemma~\ref{lemma:spanning-tree})
    \item Return $R \cup S$
\end{enumerate}
\end{algorithm}

\begin{proof}
We give the construction of a connected junction graph $J(v)$, adopting the notation introduced at the start of this section. See Algorithm~\ref{alg:junction-graph} for the corresponding pseudocode. 

Recall that a local threading $\set{x_e}$ is a set of integers satisfying the constraints specified in Definition~\ref{def:local-threading}.
Label the edges incident to vertex $v$ by the integers $1, \dots, \deg(v)$.
The goal is to form a connected graph $J(v)$ having a vertex $t_i$ for each $i \in \{1, \dots, \deg(v)\}$, where the degree of each $t_i$ is $x_i$.  We begin with an empty graph (Step~\ref{step:1}) and initialize $x_i' = x_i$ for each $i$ (Step~\ref{step:2}), where $x_i'$ represents the number of additional edges required for vertex $t_i$ to achieve its desired degree $x_i$. While $\sum_{i=1}^{d(v)} x_i' > 2(d(v)-1)$, we repeat the following steps: find the two largest values $x_{\alpha}'$ and $x_{\beta}'$ (resolving ties arbitrarily), add an edge between their corresponding vertices $t_\alpha$ and $t_{\beta}$, and decrement $x_{\alpha}'$ and $x_{\beta}'$ by $1$ (Step~\ref{step:3}). The resulting $x_i'$ values sum to $2(d(v)-1)$, and we prove below that $x_i' \geq 1$ for each $i$. Next, we construct a tree with vertex degrees $x_1',\dots,x_i'$ via the algorithm in Lemma~\ref{lemma:spanning-tree} (Step~\ref{step:4}). We return the graph that follows from these two procedures.

This graph contains no self-loops because we require $\alpha \neq \beta$ (Step~\ref{step:3b}). We further assert that the graph is connected. To prove this fact, we demonstrate the proper application of the inductive procedure outlined in the proof of Lemma~\ref{lemma:spanning-tree} in forming a tree (Step~\ref{step:4}). We only need to validate that $x_1', \dots, x_{d(v)}' \geq 1$, as $\sum_{i=1}^{d(v)}x_i' = 2(d(v)-1)$ is guaranteed upon the termination of the loop (Step~\ref{step:3}). Suppose for contradiction that $x_k' < 1$. It follows that $x_{k}' = 1$ at the start of some iteration and was subsequently decremented, either via Step~\ref{step:3a} or~\ref{step:3b}. We consider these two cases:
\begin{itemize}
    \item \textbf{Case 1} (Step~\ref{step:3a}, $k = \alpha$):
          $x_{k}' \geq x_{i}'$ for all $i \in \set{1, \dots, d(v)}$,
          so
          \[\sum_{i=1}^{d(v)} x_i' \leq d(v) \cdot x_{k}' = d(v) \leq 2(d(v)-1),\]
          a contradiction for any $d(v) > 1$, which is assumed.
    \item \textbf{Case 2} (Step~\ref{step:3b}, $k = \beta$):
          As $x_{k}' \geq x_{i}'$ for all $i \in \set{1, \dots, d(v)} \setminus \{\alpha\}$,
          so
          \[\sum_{i \in \set{1, \dots, d(v)} \setminus \set{\alpha}} x_i' \leq (d(v)-1) \cdot x_k' = d(v)-1.\]
	   Recall that
	   $\sum_{i=1}^{d(v)} x_i' = x_{\alpha}' + \sum_{i \in \set{1, \dots, d(v)} \setminus \set{\alpha}} x_i' \geq 2d(v)$
          is required to enter the loop. Hence, applying the above deduction, $x_{\alpha}' > \sum_{i \in \set{1, \dots, d(v)} \setminus \{\alpha\}} x_i'$, contradicting the below invariant (\eqn{eqn:loop-invariant}) of the loop in Step~\ref{step:3}.
\end{itemize}

\subparagraph*{Loop Invariant:}
The following invariant is maintained by the algorithm's loop (Step~\ref{step:3}), established on initialization via \cstr{c3}:
\begin{equation}
\label{eqn:loop-invariant}
x_i' \leq \sum_{j \in \set{1, \dots, d(v)} \setminus \{i\}} x_j'~\text{for all}~i \in \set{1, \dots, d(v)}
\end{equation}
We observe that $\sum_{i=1}^{d(v)}x_i$ decreases by $2$ with every iteration: either both sides of Equation~\ref{eqn:loop-invariant} are reduced by 1, thereby maintaining the inequality, or the left-hand side remains unchanged while the right-hand side is reduced by 2. In the latter scenario, counts $x_{\alpha}', x_{\beta}' \geq x_i'$ are updated in Steps~\ref{step:3a} and~\ref{step:3b}. Observe that $x_{\alpha}' \geq 2$ because $\sum_{i=1}^{d(v)} x_i' \geq 2n$ is a prerequisite for loop entry. Letting $x_i''$ denote the value of $x_i'$ at the beginning of the next iteration, we arrive at the desired conclusion:
\[x_i'' = x_i' \leq (x_{\alpha}' - 2) + x_{\beta}'
\leq \sum_{j \in \set{1, \dots, d(v)} \setminus \{i\}} x_j' - 2
= \sum_{j \in \set{1, \dots, d(v)} \setminus \{i\}} x_j''.\]

\subparagraph*{Running-Time Analysis:}
To perform Steps~\ref{step:3a} and~\ref{step:3b} efficiently, we maintain the $x'$ values in a monotone priority queue,
specifically, an array $A$ of lists $A[0], A[1], \dots, A\big[\max_{i=1}^{d(v)} x_i\big]$,
where each list $L_j$ maintains the indices $i$ for which $x'_i = j$.
We can initialize this data structure in $O\big(d(v) + \max_{i=1}^{d(v)} x_i\big)$ time,
which is $O\big(\sum_{i=1}^{d(v)} x_i\big)$ because each $x_i \geq 1$.
We also maintain the largest array index $j$ for which $A[j]$ is nonempty, and
the second-largest array index $k$ for which $A[k]$ is nonempty.
To find and decrement the maximum value in the priority queue (as in Step~\ref{step:3a}),
we extract an index $\alpha$ from list $A[j]$, decrement $x'_\alpha$,
and then append $\alpha$ to list $A[x'_\alpha] = A[j-1]$.
If $A[j]$ is now empty, we also decrement $j$; the new $A[j]$ is guaranteed to be nonempty.
To find and decrement the second-largest value in the priority queue (as in Step~\ref{step:3b}),
we extract $\beta$ from $A[j]$ if $A[j]$ has an index other than $\alpha$ (i.e., has length $>1$),
and otherwise extract from $A[k]$; then we decrement $x'_\beta$, move $\beta$ to the correct list,
and optionally decrement either $j$ or $k$ as before.
Each of these steps takes constant time, so the overall running time is
$O\big(\sum_{i=1}^{d(v)} x_i\big) = O\big(\sum_{u \in N(v)}x_{uv}\big)$.
\end{proof}

\subsection{Obtaining a Closed Walk}\label{sec:closed-walk}

Now suppose we have a junction graph $J(v)$ for every vertex~$v$,
obtained by repeatedly applying Lemma~\ref{lemma:junction-graph}
to a given local threading.
Our goal is to find a closed walk in $G$ that has no U-turns and corresponds to these junction graphs.

Define the \defn{threading graph} to be the graph whose vertices correspond to tubes and whose edges are given by the union of all junction graphs (joining at vertices corresponding to the same tube).
See Figures~\ref{fig:junction-graphs} and~\ref{fig:naive-solution} for examples.

In this threading graph, we find an \defn{Euler tour}: a closed walk that visits each edge of the graph exactly once. The presence of an Euler tour through a threading graph is guaranteed because each vertex has even degree~\cite{eulerian}, specifically twice the count $x_e$ for vertex $t_e$. The tour can be computed in time linear in the number of edges of the input graph~\cite{eulerian-alg}, which is $O(\sum_{i=1}^n x_i)$.

To ensure that U-turns are avoided in the threading, we enforce that the Euler tour does not consecutively traverse two edges of the same junction graph, which can be done in linear time by a reduction to forbidden-pattern Euler tours~\cite{Jeffrey}.

\begin{figure*}[ht]
    \centering
      \includegraphics[width=\textwidth]{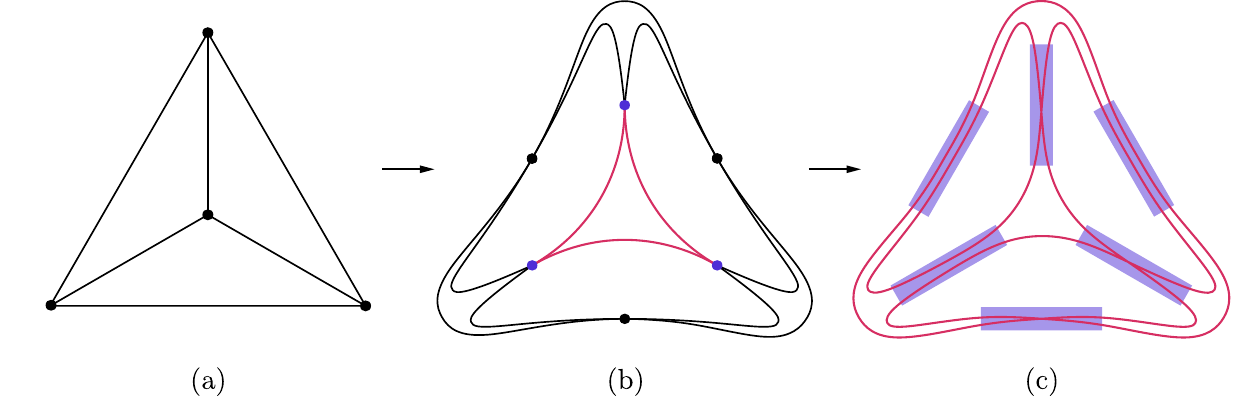}
      \caption{\label{fig:naive-solution} The target model, a threading graph featuring junction graphs that are cycles, and a threading of the input model following an Euler tour of the threading graph.}
\end{figure*}

Combining our results, we can convert a local threading $\set{x_e}$ of $G$ to a corresponding threading of $G$ in time $O\big(\sum_{v \in V}\sum_{u \in N(v)}x_{uv} + \sum_{i=1}^n x_i\big) = O\big(\sum_{i=1}^n x_i\big)$. Later in Section~\ref{sec:edge-visits}, we will show that the optimal threading satisfies $\sum_{i=1}^n x_i = O(m)$, in which case our running time simplifies to $O(m)$.

\begin{theorem}
We can convert a local threading solution of $G$ into a threading of $G$ in
$O(\sum_{i=1}^n x_i)$ time, which for an optimal threading is $O(m)$.
\end{theorem}

\section{Worst-Case Bounds}\label{sec:bounds}

In this section, we prove tight worse-case upper and lower bounds on the total length of an optimal threading (\sect{sec:edge-visits}) and on the most times one edge may be visited by an optimal threading (\sect{sec:maximum-count}).

\subsection{Total Length}\label{sec:edge-visits}

Every graph $G$ with minimum degree $\geq 2$ has a \defn{double threading} defined by assigning each junction graph $J(v)$ to be a cycle of length~$d(v)$, as depicted in \fig{fig:naive-solution}. This threading results in each tube being traversed exactly twice, totaling a length of $2m$. Thus an optimal threading has length at most $2m$. We can approach this upper bound up to an additive constant by considering graphs with long sequences of bridges, such as the graph illustrated in Figure~\ref{fig:bound-examples}a. We shall later tighten this upper bound by considering graph properties (Lemma~\ref{lemma:cycle-bound}).

Now we establish a lower bound on the total length of any threading:

\begin{lemma}\label{lm:perfectlowerbound}
    Any threading must have length at least $2m-n$.
\end{lemma}

\begin{proof}
    Each junction graph $J(v)$ is connected, so it contains at least $d(v)-1$ edges, and every edge $t_i t_j$ in $J(v)$ necessitates visits to two tubes, $t_i$ and $t_j$. By summing these visits across all junctions, we double-count visits to tubes.
    Thus, any threading $\set{x_{uv}}$ has length
    \[
    \sum_{uv\in E}x_{uv} = \frac{1}{2} \sum_{v\in V}\sum_{u\in N(v)}x_{uv}\geq \frac{1}{2}\sum_{v\in V}2(d(v) - 1) \stackbin[\textstyle {\uparrow \atop \text{\clap{(handshaking)}}}]{}{=} 2m - n.
    \]
    From the perspective of local threading, the inequality step follows from constraint \cstr{c4}.
\end{proof}

This lower bound is sometimes tight, such as in \fig{fig:junction-graphs}a, which we give a special name:
\begin{definition}\label{def:perfect-threading}
A \defn{perfect threading} is a graph threading of length $2m-n$.
\end{definition}
By the analysis in the proof of Lemma~\ref{lm:perfectlowerbound}, we obtain equivalent definitions:
\newpage
\begin{lemma} \label{lm:perfect}
    The following are equivalent for a graph threading $\set{x_{uv}}$:
    \begin{enumerate}
    \item $\set{x_{uv}}$ is a perfect threading.
    \item Every junction graph $J(v)$ is a tree, i.e., has exactly $d(v)-1$ edges.
    \item Inequality \normalfont{\cstr{c4}} holds with equality.
    \end{enumerate}
\end{lemma}

Not every graph has a perfect threading (\fig{fig:bound-examples}b). A key observation is that bridges must be threaded at least twice. If we were to remove a bridge, the graph would have two connected components, and any closed walk on the entire graph would have to enter and exit each component at least once. Because the only way to pass between the two connected components is through the bridge, the walk would have to traverse the bridge at least twice.

Hence, vertices whose incident edges are all bridges must have junction graphs containing at least $d(v)$ edges. We call these vertices \defn{London} vertices. A tighter lower bound is $2m-n+|L|$ where $L$ is the set of London vertices in $G$. Note that this bound is determined by the number of London vertices rather than the number of bridges — a London vertex connected to multiple bridges only increases the bound by 1. 

\begin{figure*}[ht]
    \centering
    \includegraphics[width=\textwidth]{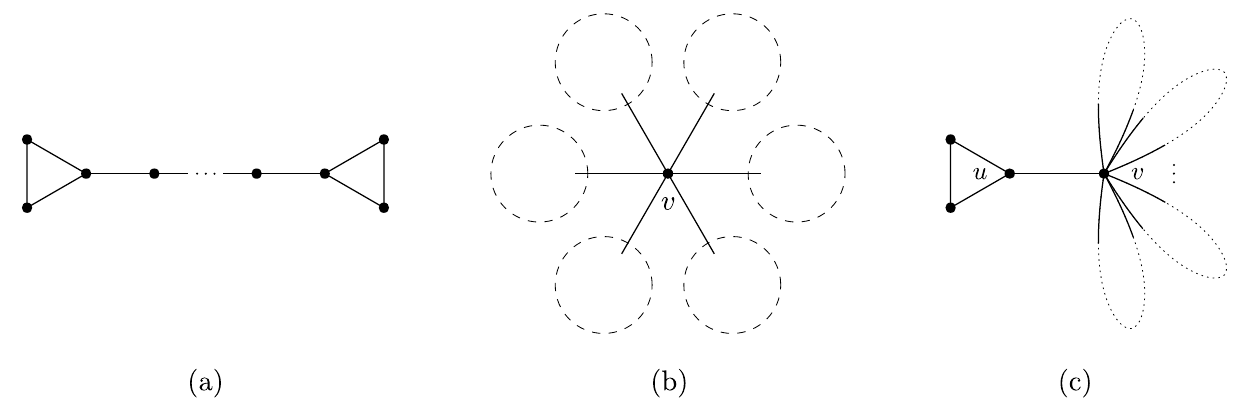}
    \caption{\label{fig:bound-examples}
    (a) A graph with a minimum threading length of $2m-6$. (b) A vertex connected to 6 disjoint parts of the graph (denoted as dashed circles). Each bridge incident to vertex $v$ is at least double-threaded, and hence $\cstr{c4}$ holds at $v$ as strict inequality, so the graph has no perfect threading. (c) The vertex $v$ has degree $\Delta$ and is connected to $\frac{\Delta - 1}{2}$ loops (dotted) of length > 5.  In an optimal threading, the edge $uv$ is threaded $\Delta-1$ times.}
\end{figure*}

Next, we consider an improved upper bound on the length of an optimal threading. While $2m$ edge visits always suffice to thread a graph, the following lemma demonstrates that this number is never necessary, as any graph without vertices of degree 1 contains a cycle.
\begin{lemma}
\label{lemma:cycle-bound}
Let $C$ be a set of vertex-disjoint simple cycles in $G$, and let $|C|$ denote the total number of edges in its cycles. In an optimal threading of $G$, at most $2m-|C|$ edge visits are needed.
\end{lemma}
\begin{proof}
We use $e \in C$ to denote edge $e$ participating in some cycle in $C$. Define the set of integers $\set{x_{e}}$ where $x_{e} = 1$ if $e \in C$ and $x_{e} = 2$, otherwise. By design, $\sum_{e\in E}x_{e} = 2m - |C|$, and so it suffices to show that $\set{x_{e}}$ is a valid threading of $G$, i.e., $\set{x_{e}}$ satisfies constraints \cstr{c1}--\cstr{c4}. Observe that each vertex $v$ is either (1) covered once by a single cycle in $C$, meaning that two of its incident edges are single-threaded while the others are threaded twice, or (2) left uncovered, in which all of its incident edges are double-threaded. In both scenarios, all constraints are clearly met. Note that \cstr{c4} holds as an equality in a vertex covered once by a cycle in $C$.
\end{proof}

In Section~\ref{sec:double-threading}, we provide an efficient algorithm for computing a threading that achieves the above bound by reduction to finding the largest set of vertex-disjoint cycles.

\subsection{Maximum Visits to One Edge}\label{sec:maximum-count}

Each edge is threaded at least once in a graph threading, but what is the maximum number of times an \textit{optimal} solution can thread an edge? In this section, we establish that no optimal threading exceeds $\Delta-1$ visits to a single edge. This upper bound is tight, as demonstrated by edge $uv$ in Figure~\ref{fig:bound-examples}c: Constraint \cstr{c4} requires multiple visits to at least one edge connected to $v$, and revisiting $uv$ is the most economical when the loops incident to $v$ are long. It is worth noting that bounding the visits to an edge by the maximum degree of its endpoints may not suffice for an optimal solution, as in the case of the left-most edge in Figure~\ref{fig:bound-examples}c, which is traversed $\frac{\Delta-1}{2} > 2$ times despite both its endpoints having a degree of 2.

\begin{lemma}\label{lm:dmax}
    An optimal threading visits a single edge at most $\Delta - 1$ times.
\end{lemma}

\begin{proof}
If $\Delta = 2$, then $G$ is a cycle, in which case the optimal threading traverses every edge once. Hence, for the remainder of this proof, we may assume $\Delta\geq 3$.

Suppose $\set{x_{e}}$ is an optimal threading of a graph $G$. Let $uv = \arg\max_{e\in E} x_{e}$ denote the edge with the highest count and assume for a contradiction that $x_{uv} \geq \Delta$. For simplicity, we first assume that $d(u),d(v) \geq 3$ and handle the case where $d(u)=2$ or $d(v)=2$ at the end. We shall show that we can remove two threads from $uv$ without violating the problem constraints. That is, the set $\set{\hat{x}_{e}}$ is a valid threading when defined as $\hat{x}_{e} = x_{uv} - 2$ if $e = {uv}$ and $\hat{x}_{e}= x_{e}$, otherwise. This conclusion contradicts our assumption that $\set{x_{e}}$ is optimal. The key to this  proof is the following:

\subparagraph*{\cstr{c4}:} Because $\set{x_e}$ satisfies \cstr{c3}, $\sum_{i=1}^{d(v)-1} x_{u_iv}\geq x_{uv} \geq \Delta$, and so
\[
\sum_{w\in N(v)}\hat{x}_{wv} = \hat{x}_{uv} + \sum_{i=1}^{d(v)-1} x_{u_iv}  \geq (\Delta-2) + \Delta \geq 2(d(v) -1).
\]
By symmetry, $u$ also satisfies \cstr{c4}, and therefore \cstr{c4} is met by all vertices of $G$. We are left to show that $\set{\hat{x}_{e}}$ satisfies \cstr{c1}--\cstr{c3}.

\subparagraph*{(C1):} $\hat{x}_{uv} > \Delta-2 \geq 1$. For any other edge $\hat{x}_e = x_e \geq 1$.

\subparagraph*{(C2):} Constraint \cstr{c2} is met as we do not modify the parity of any count.

\subparagraph*{(C3):} We now show \cstr{c3} is satisfied for $v$ and by symmetry, $u$, and therefore met by all vertices of $G$. Let us denote the neighbors of
$v$ by $u, u_1, \dots, u_{d(v)-1}$. We have
\[
\sum_{w\in N(v) \setminus \set{u}}\hat{x}_{wv}
= \sum_{w\in N(v) \setminus \set{u}}x_{wv}
\geq x_{uv}
> \hat{x}_{uv},
\]
so \cstr{c3} is satisfied for $uv$. We now demonstrate \cstr{c3} also holds for the remaining $u_iv$'s.
If $d(v) \geq 4$, because $x_{uv}\geq x_{u_iv}=\hat{x}_{u_iv}$ by our choice of $uv$, we have
\[
\sum_{w\in N(v)\setminus \set{u_i}}\hat{x}_{wv}\underset{\cstr{c1}}{\geq}
\hat{x}_{uv} + \underbrace{d(v) - 2}_{\geq 2}
\geq (x_{uv} - 2) + 2 = x_{uv} \geq  \hat{x}_{u_iv},
\]
as desired. Otherwise, $d(v) = 3$. Without loss of generality, we want to show that
\[
x_{u_1v} \leq
\hat{x}_{uv} + \hat{x}_{u_2v}
= x_{uv} + x_{u_2v} - 2.
\]
Because $x_{uv}\geq x_{u_1v}$ (by choice of $uv$) and $x_{u_2v}\geq 1$ (from \cstr{c1}), this inequality holds in all cases except when $x_{u_1v}=x_{uv}$ and $x_{u_2v}=1$. However, in this particular scenario, the sum of counts surrounding $v$ amounts to $2x_{uv}+1$, which contradicts \cstr{c2}.

If either endpoint of $uv$ has degree $2$, then we instead consider the maximal path $w_1, \ldots, w_\ell$ including $uv$ such that all intermediate vertices have degree~$2$: $d(w_2) = \ldots = d(w_{\ell - 1}) = 2$.
Thus, $d(w_1), d(w_{\ell}) \geq 3$ (as we are in the case $\Delta \geq 3$) and $u v = w_i w_{i+1}$ for some~$i$.
Because $\set{x_e}$ is a valid threading, we must have $x_{w_1w_2} = \cdots = x_{w_{\ell-1}w_\ell}=x_{uv}\geq \Delta$.
Now we modify the threading $\set{x_e}$ by removing two threads from each $x_{w_i w_{i+1}}$ to obtain $\set{\hat{x}_e}$. Constraints \cstr{c1}--\cstr{c4} remain satisfied at the degree-$2$ vertices $w_2, \ldots, w_{\ell-1}$. Finally, we can apply the proof above to show that the constraints remain satisfied at the end vertices $w_1$ and $w_\ell$ of degree at least~$3$.
\end{proof}

\section{Polynomial-Time Algorithm via Perfect Matching}\label{sec:algorithm}

In this section, we present our main result: a polynomial-time algorithm for computing an optimal threading of an input graph $G$. Our approach involves reducing \threading~to the problem of min-weight perfect matching, defined as follows.

A \defn{matching} in a graph is a set of edges without common vertices. A \defn{perfect matching} is a matching that covers all vertices of the graph, i.e., a matching of cardinality $\frac{n}{2}$. If the graph has edge weights, the \defn{weight} of a matching is the sum of the weights of its edges, and a \defn{min-weight perfect matching} is a perfect matching of minimum possible weight.

We begin by constructing a graph that possesses a perfect matching if and only if $G$ has a \emph{perfect} threading (Definition~\ref{def:perfect-threading}). This construction gives a reduction from determining the existence of a perfect threading to the perfect matching problem. Next, we extend this construction to ensure perfect matching always exists. In this extended construction, a perfect matching of weight $W$ corresponds to a threading of length $W + m$, giving a reduction from \threading\ to finding a min-weight perfect matching.

\subsection{Determining Existence of a Perfect Threading} \label{sct:perfectthread}
By Lemma~\ref{lm:perfect}, a threading $\set{x_{uv}}$ of a graph $G$ is a perfect threading if and only if it satisfies inequality \cstr{c4} with equality:

\begin{enumerate}[itemindent=.65cm,label=\textbf{(C$^*$\arabic*)},ref=C$^*$\arabic*]
    \setcounter{enumi}{3} 
    \item \label{c4s} $\sum_{u \in N(v)} x_{uv} = 2(d(v) - 1)$ for all $ v \in V$.
\end{enumerate}
In fact, most of the other constraints become redundant in this case:

\begin{lemma} \label{lm:perfect-c1-c4*}
    $\set{x_{uv}}$ is a perfect threading if and only if it satisfies \normalfont{\cstr{c1}} and \normalfont{\cfourstar}.
\end{lemma}

\begin{proof}
    If $\set{x_{uv}}$ satisfies \cfourstar, then it satisfies constraint \cstr{c2}, because $2(d(v)-1)\equiv 0 \pmod{2}$.
    \cfourstar\ can be rewritten as
    $x_{uv} + \sum_{w \in N(v) \setminus \set{u}} x_{w v} = 2(d(v) - 1)$,
    and by \cstr{c1},
    $\sum_{w \in N(v) \setminus \set{u}}x_{wv} \geq d(v) - 1$,
    so \cstr{c3} also holds.
\end{proof}

Consider a vertex $v$ and its neighbors $u_1, \ldots, u_{d(v)}$. We can think of constraint \cfourstar~as allocating $2(d(v)-1)$ units among $x_{u_1v}, \ldots, x_{u_{d(v)}v}$. First, we must allocate one unit to each $x_{u_iv}$ in order to satisfy \cstr{c1}. This leaves $d(v) - 2$ units to distribute among the edges.

We show how to simulate this distribution problem by constructing a graph $H$ that has a perfect matching if and only if, for every vertex $v$, we can distribute $d(v)-2$ units among its neighboring $x_{u_i v}$. Thus $H$ has a perfect matching if and only if $G$ has a perfect threading.

Given a graph $G$, define the graph $H$ as follows; refer to \fig{fig:matching-reduction}. For each edge $uv \in E(G)$, create a perfect matching of $d_{uv} := \min\{d(u), d(v)\} - 2$ disjoint edges $(\overline{u}v_i,u\overline{v}_i)$, among $2 \, d_{uv}$ created vertices $\overline{u}v_1, u\overline{v}_1, \ldots, \overline{u}v_{d_{uv}}, u\overline{v}_{d_{uv}}$.%
\footnote{In the same way that $uv$ and $vu$ denote the same edge, we treat labels $u \overline{v}$ and $\overline{v} u$ to be equivalent. Thus, the notation $u\bar{v}_i$ and $\bar{v}u_i$ refers to the same vertex.}
For each vertex $v$, create $d(v) - 2$ vertices labeled $v_1, \ldots, v_{d(v)-2}$. For every edge $uv$ incident to $v$, add an edge between vertices $v_i$ and $u\overline{v}_j$ for all $1 \leq i \leq d(v) - 2$ and $1 \leq j \leq d_{uv}$ (forming a biclique). Note that any vertex of degree 2 disappears in this construction because of the $-2$ in each creation count.

\begin{figure*}[ht]
\centering
  \includegraphics[width=\linewidth]{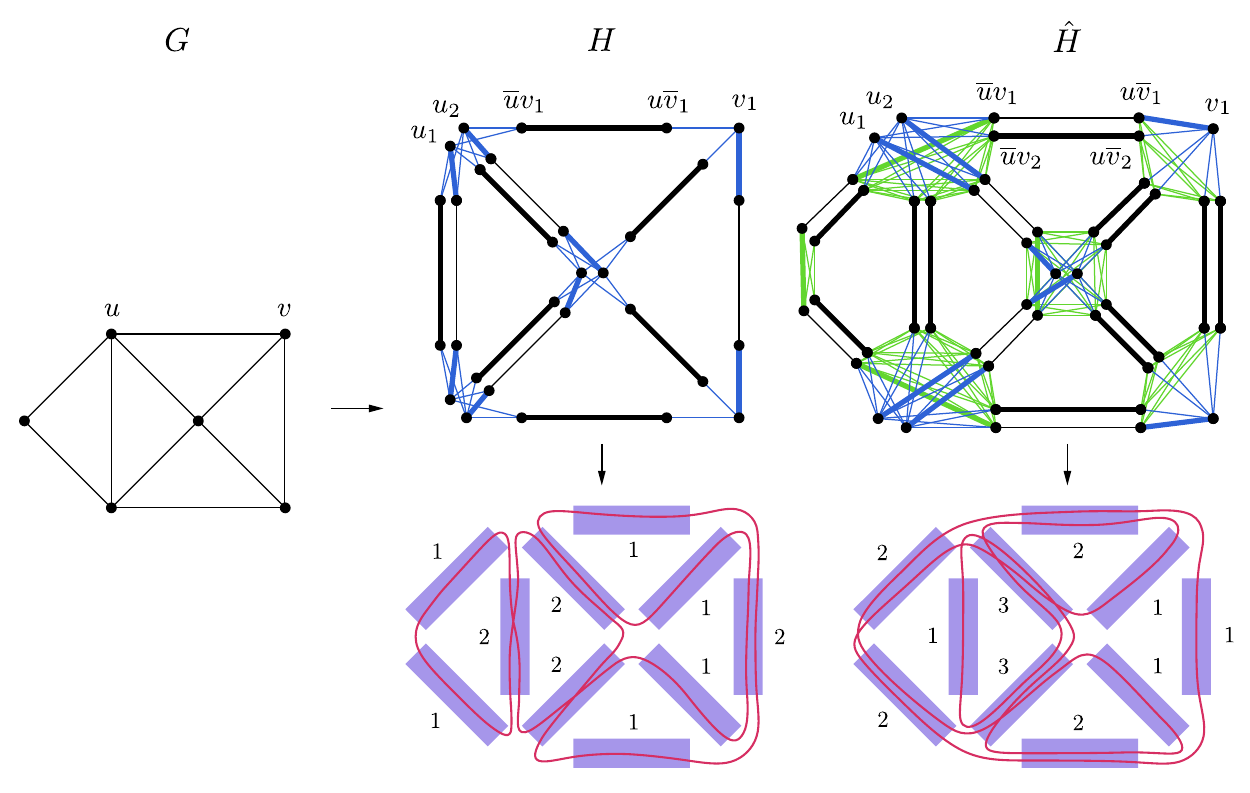}
  \caption{\label{fig:matching-reduction} Construction of $H$ and $\hat{H}$ from $G$, each with some matching in bold and a corresponding threading to the matching labeled with counts.}
\end{figure*}

\begin{theorem}\label{thm:perfectgtoh}
    $G$ has a perfect threading if and only if $H$ has a perfect matching.
\end{theorem}

To prove Theorem~\ref{thm:perfectgtoh}, we will show how to translate between a perfect threading of $G$ and a perfect matching of $H$.
Given a matching $M \subset E(H)$ of $H$, define a possible threading solution $\varphi(M) = \set{x_{uv}}$ by taking $x_{uv}$ to be $1$ plus the number of edges $(\overline{u}v_i,u\overline{v}_i)$ that are \textit{not} included in $M$:
$x_{uv} := 1 + \big|{\set{(\overline{u}v_i, u\overline{v}_i) : 1 \leq i \leq d_{uv}} \setminus M}\big|$.

\begin{claim}\label{clm:perfectthread}
    If $M$ is a perfect matching in $H$, then $\varphi(M)$ is a perfect threading of $G$.
\end{claim}

\begin{proof}
    By Lemma~\ref{lm:perfect-c1-c4*}, it suffices to prove that $\varphi(M)$ satisfies \cstr{c1} and \cfourstar.
    The $1+$ in the definition of $\varphi(M)$ satisfies \cstr{c1}.
    For every vertex $v \in V$, the vertices $v_1, \ldots, v_{d(v) - 2}$ are all matched to vertices of the form $u\overline{v}_i$; for each such matching pair, the edge $(u\overline{v}_i, \overline{u}v_i) \notin M$. Conversely, for any vertex $u\overline{v}_i$ that is not matched to any $v_j$, the edge $(u\overline{v}_i, \overline{u}v_i)$ must be part of the matching. Hence, for each vertex $v$, the number of edges of the form $(u\overline{v}_i, \overline{u}v_i)$ that are not included in $M$ is exactly $d(v) - 2$. The sum $\sum_{u \in N(v)} x_{uv}$ includes this count and $d(v)$ additional $1$s, so equals $(d(v) - 2) + d(v) = 2(d(v)-1)$, satisfying \cfourstar.
\end{proof}

\begin{claim}\label{clm:perfectthreadreverse}
    For any perfect threading $\set{x_{uv}}$ of $G$, there exists a perfect matching $M$ of $H$ such that $\varphi(M) = \set{x_{uv}}$.
\end{claim}

\begin{proof}
    Given a perfect threading $\set{x_{uv}}$ of $G$, we construct a perfect matching of $H$ as follows.
    First, for every $uv\in E(G)$, we match the edges $(\overline{u}v_1, u\overline{v}_1), \ldots, (\overline{u}v_{d_{uv}-x_{uv}+1}, u\overline{v}_{d_{uv}-x_{uv}+1})$.
    We show that index $d_{uv}-x_{uv}+1$ is always non-negative; when it is zero, we match no such edges.
    By constraint \cfourstar, $x_{uv} = 2(d(v) - 1) - \sum_{w \in N(v) \setminus \set{u}} x_{wv}$.
    By constraint \cstr{c1}, each term in the sum is at least $1$, so $x_{uv}\leq d(v) - 1$. Thus $x_{uv}\leq d_{uv} + 1$, i.e., $d_{uv}-x_{uv}+1 \geq 0$.

    With our matching so far, the number of unmatched vertices of the form $u\overline{v}_i$ at each vertex $v$ is
    $\sum_{u \in N(v)} (x_{uv}-1)$.
    By \cfourstar, this count is exactly $2 (d(v) - 1) - d(v) = d(v) - 2$.
    Thus we can match each of these unmatched vertices to a unique vertex $v_j$ to complete our perfect matching.
\end{proof}

Claims~\ref{clm:perfectthread} and~\ref{clm:perfectthreadreverse} complete the proof of Theorem~\ref{thm:perfectgtoh}.

\subsubsection{Running-Time Analysis}

First, let us calculate the sizes of $V(H)$ and $E(H)$. Recall that $H$ has $d(v) - 2$ vertices corresponding to every vertex $v\in V(G)$, and up to $2(\min\{d(u), d(v)\} - 2)\leq 2\Delta$ vertices corresponding to every edge $uv\in E(G)$. Therefore, the maximum number of vertices in $H$ is
\[
\sum_{v\in V}(d(v) - 2) + 2 \sum_{uv\in E}\Delta \leq  2m - 2n + 2m\Delta = O( m \Delta).
\]
Now recall that $H$ has $\min\{d(u), d(v)\} - 2\leq \Delta$ edges for every $uv$ and at most $\Delta^3$ edges for every $v$. Thus, the total number of edges in $H$ is upper-bounded by 
\[
2\sum_{uv\in E}\Delta + \sum_{v\in V}\Delta^3\leq 2m\cdot \Delta + n \Delta^3 = O(n\Delta^3).
\]
We conclude that $H$ can be constructed in $O(n \Delta^3 + m \Delta)$ time. 

Micali and Vazirani \cite{perfectmatching} gave an algorithm that computes the maximum matching of a general graph in $O(\sqrt{n}m)$ time, thereby enabling us to verify the existence of a perfect matching. It follows that we can determine a perfect matching of $H$ in time
\[
O(\sqrt{|V(H)|}\cdot |E(H)|) = O(\sqrt{m\Delta}\cdot n\Delta^3) = O(n\sqrt{m}\cdot \Delta^{3.5}).
\]
This running time exceeds the construction time of $H$, and so it is the final running time of our algorithm. 

Note that we can improve the bound on the size of $H$ by considering the  \textit{arboricity} of $G$. The arboricity of a graph $\alpha(G)$ is defined as the minimum number of edge-disjoint spanning forests into which $G$ can be decomposed \cite{chibanishizeki}. This parameter is closely related to the degeneracy of the graph and is often smaller than $\Delta$. Chiba and Nishizeki~\cite{chibanishizeki} show that $\sum_{uv\in E}\min\{d(u), d(v)\} \leq 2m\alpha(G)$, which would give us a tighter bound on the size of $V(H)$.

In summary, we can find a perfect threading of $G$, if one exists, by determining a perfect matching in $H$ in $O(n\sqrt{m}\cdot \Delta^{3.5})$ time.

\subsection{Finding an Optimal Threading}\label{sct:optthreading}

Now we examine the general scenario where a perfect threading may not exist, i.e., \cstr{c4} may hold with a strict inequality for some vertex. The graph $H$ constructed in \sect{sct:perfectthread} permits exactly $2(d(v) -1)$ visits to vertex $v$. We aim to allow more visits to $v$ while satisfying constraints \cstr{c2} and \cstr{c3}.

In a general threading, $x_{uv}\leq \min\{d(u), d(v)\} - 1$ (as argued in Claim~\ref{clm:perfectthreadreverse}) is not necessarily true. However, Lemma~\ref{lm:dmax} gives us a weaker upper bound, $x_{uv}\leq \Delta - 1$, for any optimal threading. We therefore modify the construction from \sect{sct:perfectthread} in two ways. First, we generate $\Delta-2$ copies of every edge, regardless of the degree of its endpoints. Second, for every pair of edges $uv$ and $wv$ meeting at vertex $v$, we introduce an edge between $u\overline{v}_i$ and $w\overline{v}_j$ for all $1 \leq i,j\leq \Delta-2$. Intuitively, these edges represent threads passing through $v$, going from  $uv$ to $wv$, after having met the lower bound of $2(d(v)-1)$ visits.

More formally, we define a weighted graph $\hat{H}$ from $G$ as follows; refer to \fig{fig:matching-reduction}.
For each edge $uv \in E(G)$, create a weight-$0$ perfect matching of $\Delta - 2$ disjoint weight-$0$ edges $(\overline{u}v_i,u\overline{v}_i)$, among $2 (\Delta-2)$ created vertices $\overline{u}v_1, u\overline{v}_1, \ldots, \overline{u}v_{\Delta-2}, u\overline{v}_{\Delta-2}$;
these edges are black in \fig{fig:matching-reduction}.

For every vertex $v$, create $d(v)-2$ vertices $v_1, \ldots, v_{d(v)-2}$, and add a weight-$\frac12$ edge $(v_i, u\overline{v}_j)$ for every $u\in N(v)$ and $1 \leq i\leq d(v) - 2, j\leq \Delta - 2$;
these edges are blue in \fig{fig:matching-reduction}.
Finally, for each pair of edges $uv$ and $wv$ incident to $v$, create a weight-$1$ edge $(u\overline{v}_i, w\overline{v}_j)$ for every $1 \leq i,j\leq \Delta-2$;
these edges are green in \fig{fig:matching-reduction}.

\begin{theorem}\label{thm:optimalgtoh}
    $G$ has a threading of length $W + m$ with $\max_{uv\in E(G)}x_{uv}\leq \Delta - 1$ if and only if $\hat{H}$ has a perfect matching of weight $W$.
\end{theorem}

To prove Theorem~\ref{thm:optimalgtoh}, we again show how to translate between a threading of $G$ and a perfect matching of $\hat{H}$.
Given a matching $M \subset E(\hat H)$ of $\hat H$, define a possible threading solution $\psi(M) = \set{x_{uv}}$ by taking $x_{uv}$ to be $1$ plus the number of copies of $uv$ not matched in~$M$:
$x_{uv} := 1 + \big|{\set{(\overline{u}v_i, u\overline{v}_i) : 1 \leq i \leq \Delta-2} \setminus M}\big|$.

\begin{claim} \label{lm:matching-to-threading}
    If $M$ is a perfect matching in $\hat{H}$ of weight $W$, then $\psi(M) = \set{x_{uv}}$ is a threading of $G$ of length $W + m$ with $\max_{uv\in E(G)}x_{uv}\leq \Delta - 1$.
\end{claim}

\begin{proof}
By definition of $\psi(M)$, every $x_{uv}$ satisfies $1\leq x_{uv}\leq \Delta - 1$. Thus, $\set{x_{uv}}$ satisfies \cstr{c1} and $\max_{uv \in E(G)} x_{uv}\leq \Delta - 1$.

Let $a_{v}(uv)$ denote the number of vertices $u\overline{v}_i$ (for $1 \leq i \leq \Delta-2$) matched with some vertex $v_j$, i.e., the number of blue edges incident to a vertex $u\overline{v}_i$ that appear in~$M$.
Let $b_{v}(uv)$ denote the number of vertices $u\overline{v}_i$ (for $1 \leq i \leq \Delta-2$) matched with some vertex $w\overline{v}_j$, i.e., the number of green edges incident to a vertex $u\overline{v}_i$ that appear in~$M$.
Any other vertex $u\overline{v}_i$ (not incident to either a blue or green edge in~$M$) must be matched to its corresponding vertex $\overline{u}v_i$, which does not contribute to $x_{uv}$. Hence, $x_{uv} = 1 + a_v(uv) + b_v(uv)$.

Next we prove that $\set{x_{uv}}$ satisfies constraint \cstr{c4}.
For every vertex $v$, we have $\sum_{u\in N(v)} a_v(uv) = d(v) -2$, which implies $\sum_{u\in N(v)}(x_{uv}-1)\geq d(v) -2$, which is equivalent to \cstr{c4}.

Next consider \cstr{c2}. Any edge $(u\overline{v}_i, w\overline{v}_j)$ present in $M$ adds $1$ to both $b_v(uv)$ and $b_v(wv)$, thereby ensuring $\sum_{u\in N(v)}b_v(uv) \equiv 0 {\pmod{2}}$. Consequently,
\[
\sum_{u\in N(v)}x_{uv} \equiv \sum_{u\in N(v)}(a_v(uv)+1)  = 2(d(v) -1) \equiv 0 \pmod 2.
\]

Finally, consider \cstr{c3}.
Given that $a_v(uv)\leq d(v)-2$, we infer $\sum_{w\in N(v)\setminus \set{u}}a_v(uv) + d(v) -1 \geq a_v(uv) + 1$. Additionally, for each vertex contributing to $b_v(uv)$, its matched vertex contributes to some $b_v(wv)$, so $\sum_{w\in N(v)\setminus \set{u}} b_v(wv)\geq b_v(uv)$. Hence, we have
\[
\sum_{w\in N(v)\setminus \set{u}} x_{wv} = \sum_{w\in N(v)\setminus \set{u}} (a_v(wv) + b_v(wv) + 1)\geq (a_v(uv) + 1) + b_v(uv) = x_{uv}.
\]
We conclude that $\set{x_{uv}}$ is a threading of $G$.

Lastly, we compute its length. The weight of $M$ is determined by the number of blue and green edges it contains because the edges $(\overline{u}v_i, u\overline{v}_i)$ have zero weight. Each of its blue edges of the form $(v_i, u\overline{v}_j)$ has weight $\frac12$ and is accounted for once in $a_v(uv)$, for a total weight of $a_v(uv)/2$. Each of its green edges of the form $(u\overline{v}_i, w\overline{v}_j)$ has weight $1$ and is counted twice --- once in $b_v(uv)$ and once more in $b_v(wv)$ --- for a total weight of $b_v(uv)/2$. Hence, the weight $W$ of the matching $M$ is given by
\[
W = \sum_{v\in V}\sum_{u\in N(v)} \left(\frac{a_v(uv)}{2} + \frac{b_v(uv)}{2}\right) = 2 \sum_{uv\in E}\frac{x_{uv}-1}{2} = \sum_{uv\in E} x_{uv}-m.
\]
Therefore $\set{x_{uv}}$ is a threading of $G$ of length $W+m$.

\end{proof}

\begin{claim}
    For every threading $\set{x_{uv}}$ of $G$ such that $\max_{uv \in E(G)} x_{uv}\leq \Delta-1$, $\hat{H}$ has  a perfect matching $M$ such that $\psi(M)= \set{x_{uv}}$.
\end{claim}
\begin{proof}
Let $\set{x_{uv}}$ be a threading of $G$ satisfying $x_{uv}\leq \Delta - 1$ for every edge $uv\in E$. Recall Lemma~\ref{lemma:junction-graph}, where we demonstrate the construction of a junction graph $J(v)$ for vertex $v$.

For every vertex $v \in V$, we know by \cstr{c2} and \cstr{c4} that $\sum_{u\in N(v)}x_{uv} = 2(d(v)-1) + 2k$ for some integer $k$. Note that $J(v)$ has $d(v)$ vertices and $d(v) -1 + k$ edges. Because $J(v)$ is connected, we can thus select $k$ edges from $J(v)$ such that removing them will leave behind a tree. Denote these edges by $(u^1, w^1),\ldots, (u^k, w^k)$ where $u^1, \ldots, u^k, w^1, \ldots, w^k \in N(v)$. For each edge $(u^\ell, w^\ell)$, match a green edge of the form $(u^\ell\overline{v}_i, w^\ell \overline{v}_j)$. For every edge $uv$ connected to $v$, denote by $b_v(uv)$ the number of vertices of the form $u\overline{v}_i$ currently matched, i.e., the number of times $u$ appears as an endpoint among the $k$ edges selected from $J(v)$. 

Because the edges remaining in $J(v)$ after removing $(u^1, w^1),\ldots, (u^k, w^k)$ form a tree, every neighbor of $v$ must have at least one incident edge in $J(v)$ that is \textit{not} selected. Because the degree of $t_{uv}$ in $J(v)$ is $x_{uv}$, the number of matched vertices must satisfy $b_v(uv) \leq x_{uv} - 1$.\footnote{Here $t_{uv}$ is a vertex representing the tube $uv$. See the notation in Section \ref{sec:junction-graph}.}

For each $u\in N(v)$, let $a_{v}(uv) = x_{uv} - b_v(uv) - 1$. It is clear from our above observation that $a_v(uv) \geq 0$. Given $\sum_{u\in N(v)}b_v(uv) = 2k$, we have $\sum_{u\in N(v)}a_{v}(uv) = d(v) - 2$. It follows that we can match $a_{v}(uv)$ vertices in $u\overline{v}_1, \ldots, u\overline{v}_{\Delta - 2}$ to an equal number of vertices in $v_1, \ldots, v_{d(v)-2}$ using blue edges. After executing this procedure, all vertices of the form $v_1, \ldots, v_{d(v)-2}$ will have been matched. Furthermore, the number of matched vertices of the form $u\overline{v}_i$ is exactly $a_v(uv) + b_v(uv) = x_{uv} - 1$. We repeat this procedure for all vertices.

Now, for every edge $uv$, there are two sets of unmatched vertices, each of size $\Delta - 2 - (x_{uv} - 1) = \Delta - x_{uv} - 1$  $u\overline{v}_i$, of the form $u\overline{v}_i$ and $\overline{u}v_j$, respectively. By rearranging the existing matches, we can ensure these vertices are exactly $u\overline{v}_1, \ldots, u\overline{v}_{\Delta - x_{uv} - 1}, \overline{u}v_1, \ldots, \overline{u}v_{\Delta - x_{uv}-1}$. Then we can proceed to match every pair $(u\overline{v}_i, \overline{u}v_i)$, for $i\leq \Delta - x_{uv} - 1$, using a black edge.

The above process results in a perfect matching $M$ from the threading $\set{x_{uv}}$. The number of edges of the form $(u\overline{v}_i, \overline{u}v_i)$ included in the matching is precisely $\Delta - x_{uv} - 1$. Hence, $\psi(M) = \set{x_{uv}}$.
\end{proof}

The above two claims complete the proof of Theorem~\ref{thm:optimalgtoh}.
Lemma~\ref{lm:dmax} establishes that an optimal threading visits an edge no more than $\Delta - 1$ times, so $\hat{H}$ must have a perfect matching. Furthermore, if $M$ is the min-weight perfect matching of $\hat{H}$, then $\psi(M)$ is the optimal threading of $G$. We can therefore find the optimal threading of $G$ by finding the min-weight perfect matching of $\hat{H}$ and applying the reduction of Claim~\ref{lm:matching-to-threading}. 

Note that the solution presented in this section can be readily adapted to address a constrained variant of \threading, where each edge is allowed to be traversed only a limited number of times by imposing limits on the number of vertex and edge copies created during the construction of $\hat{H}$. This scenario arises, for example, when dealing with tubes of restricted diameter.

\subsubsection{Running-Time Analysis}
\label{app:optimal-runtime}

First, let us analyze the size of $\hat{H}$: the graph contains $\Delta - 2$ vertices for each vertex $v\in V(G)$ and $2(\Delta - 2)$ vertices for each edge $uv\in E(G)$. Hence, the total number of vertices in $\hat{H}$ is $O(m\Delta)$. In terms of edges, $\hat{H}$ includes $ \Delta -2$ edges for each edge $uv\in E(G)$ and no more than $\Delta^4$ edges for each vertex $v\in V(G)$. Therefore, the total edge count in $\hat{H}$ is $O(n\Delta ^4)$. As a result, the construction of $\hat{H}$ requires $O(m\Delta + n\Delta^4)$ time.

Next, we use the algorithm of Galil, Mical, and Gabow~\cite{minweightperfectmatching} to find a minimum weight perfect matching of $\hat{H}$. This algorithm has time complexity $O(nm \log n)$, and so on $\hat{H}$ it runs in time 
\[
O(|V(H)||E(H)|\log (|V(H)|)) = O(m\Delta\cdot n\Delta^4 \cdot \log (m\Delta)) = O(nm\cdot \Delta^5\log n).
\]
As this term dominates the time for constructing $\hat{H}$, we conclude that our algorithm for \threading~runs in time $O(nm \cdot \Delta^5 \log n)$.

\subsubsection{Extension to Weighted Graphs}
\label{sec:weighted}

In this section, we adapt our \threading\ algorithm to weighted graphs representing structures whose edges have varying lengths. Specifically, we introduce a weight function $\ell: E \to \mathbb{R}^{+}$, where $\ell(e)$ represents the length of tube $e$. The goal of \threading\ is now to minimize the \defn{total length} of a threading $T$, defined as $\sum_{e \in T} \ell (e)$. This problem is equivalent to the weighted version of \ip\ where we seek to minimize $\sum_{e\in E} \ell (e) \, x_e$ subject to constraints \cstr{c1}--\cstr{c4}. 

Our \threading\ algorithm hinges upon Lemma~\ref{lm:dmax}. Fortunately, this result holds for weighted graphs. In the proof of the lemma, we demonstrated that if any threading $\set{x_e}$ has $x_e \geq \Delta$ for some $e \in E$, then we can construct a strictly shorter threading $\set{x'_e}$ that remains consistent with constraints \cstr{c1}--\cstr{c4}. Specifically, $x'_e \leq x_e$ for all $e \in E$ and $x'_e < x_e$ for at least one $e \in E$. Therefore, even in the weighted case we have $\sum_{e\in E} \ell(e) \, x'_e < \sum_{e\in E} \ell(e) \, x_e$ for any weight function $\ell: E \to \mathbb{R}^{+}$. Hence, an optimal threading never traverses an edge more than $\Delta - 1$ times as desired. 

To adapt our \threading\ algorithm for to the weighted scenario, we construct a graph similar to $\hat{H}$ in \sect{sct:optthreading}, but with modified edge weights: a blue edge $(v_i, u\bar{v}_j)$ now has weight $\frac{1}{2}\ell(uv)$ instead of weight $\frac{1}{2}$, and a green edge $(u\bar{v}_i, w\bar{v}_j)$ has weight $\frac{1}{2} \big(\ell(uv) + \ell(wv)\big)$ rather than weight $1$. The black edges continue to have zero weight.
Denote this new graph by $\Tilde{H}$.

By a similar proof to that of \autoref{thm:optimalgtoh}, we obtain a reduction from weighted \threading\ to minimum-weight perfect matching:

\begin{theorem}\label{thm:weightoptimalgtoh}
    $G$ has a threading of length $W + \sum_{e\in E(G)}\ell(e)$ with $\max_{e\in E(G)}x_{e}\leq \Delta - 1$ if and only if $\Tilde{H}$ has a perfect matching of weight $W$.
\end{theorem}

As before, an edge $uv$ traversed by a threading corresponds to an edge $(u\bar{v}_i,\bar{u}v_i)$ that is \emph{not} part of the perfect matching of $\Tilde{H}$. Both endpoints of this edge must be matched with either a green or blue edge. Each such matching contributes $\frac{\ell(uv)}{2}$ to the matching's total weight. Thus, we can show that a perfect matching in $\Tilde{H}$ with weight $W$ corresponds to a threading of $G$ of length $W + \sum_{e\in E} \ell(e)$.

\section{Special Cases}
\label{sec:special}

Here we focus on two scenarios: \threading~on cubic graphs and \double, where each edge can be traversed at most twice. 

\subsection{Cubic Graphs}\label{sec:cubic-graphs}

If graph $G$ is cubic, then by Lemma~\ref{lm:dmax}, an optimal threading of $G$ visits each edge at most twice. Furthermore, in a perfect threading of $G$, if it exists, exactly one edge incident to each vertex is double-threaded due to constraint \cfourstar. Hence, it follows that $G$ has a perfect threading if and only if $G$ has a perfect matching. A perfect matching of $G$ gives the set of edges to be double-threaded in a perfect threading. Every bridgeless cubic graph has a perfect matching~\cite{cubic-bridgeless}---it can be computed in $O(n\log^4{n})$ time~\cite{matching-cubic}. In fact, if all bridges of a connected cubic graph $G$ lie on a single path of $G$, then $G$ has a perfect matching~\cite{errera}.

\subsection{Double Threading}\label{sec:double-threading}

In \double, the goal is to minimize the number of double-threaded edges or, equivalently, to maximize the number of edges visited only once. A solution to \double~on a cubic graph also solves \threading~on the same graph. This is due to the observation that either zero or two single-threaded edges are incident to each vertex in a solution to \double, which aligns with the reality of \threading~on cubic graphs. By the same observation, a solution to \double~matches the upper bound given in Lemma~\ref{lemma:cycle-bound} for general graphs. We further note that \double~may be reduced to the task of finding vertex-disjoint cycles simple with maximum collective length, which we solve below in Algorithm~\ref{alg:vertex-disjoint-cycles}.

\begin{algorithm}
\caption{Maximum Length Vertex-Disjoint Simple Cycles}
\label{alg:vertex-disjoint-cycles}
\begin{enumerate}
    \item
    Construct a weighted graph $G'$ from $G$
    (Figure~\ref{fig:double-threading-reduction}):
    \begin{enumerate}
        \item For each vertex $v \in V$, create a complete bipartite graph $G_u = K_{d(v),d(v)}$ with zero-weight edges. Let $D_u^-$ and $D_u^+$ denote the two disjoint vertex sets of this graph.
        \item For each edge $uv \in E$, add an edge unit weight between a vertex of $D_u^+$ and a vertex of $D_v^+$ such that each vertex of $D_u^+$ and $D_v^+$ has exactly one edge incident to it.
        \item 
        \label{step:edge}
        For each subgraph $G_v$, add a zero-weight edge between any two vertices of $D_v^-$.
    \end{enumerate}
    \item 
    \label{step:matching}
    Compute a maximum weight perfect matching $M$ in $G'$.
    \item Return edge set $S \subset E$ of $G$ corresponding to the weighted edges of $M$. 
\end{enumerate}
\end{algorithm}
\begin{figure*}[ht]
\centering
  \includegraphics[width=\linewidth]{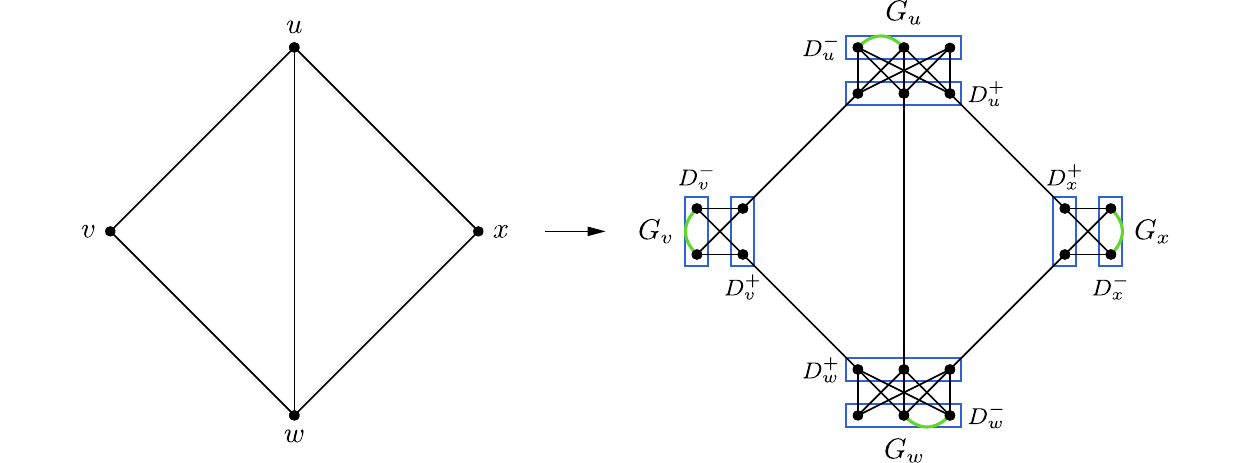}
  \caption{\label{fig:double-threading-reduction}Illustration of constructing $G'$ from $G$.}
\end{figure*}

We sketch the intuition behind why matching $M$ corresponds one-to-one to vertex disjoint simple cycles in $G$. Observe two cases for each vertex $u \in V(G)$: (i) If $M$ contains the edge of \ref{step:edge}, then $d(u)-2$ vertices in $D_{u}^{-}$ match with vertices in $D_{u}^{+}$, leaving two vertices in $D_{u}^{+}$ to match with their neighbors in adjacent subgraphs; (ii) all vertices in $D_{u}^{+}$ are saturated via connections to $D_{u}^{-}$. Hence, each vertex $u$ is in exactly one cycle (i) or none at all (ii).

\subparagraph*{Running-Time Analysis:}
We begin our analysis of the running time of Algorithm~\ref{alg:vertex-disjoint-cycles} by first bounding the size of $G'$. Each subgraph $G_v$ has $2d(v)$ vertices and $d(v)^2 + 1$ edges, and these subgraphs are connected via $m$ edges. Because $\sum_{v \in V} d(v) = 2m$ and $\sum_{v \in V} d(v)^2 \leq m(2m/(n-1)+n-2)$~\cite{DECAEN1998245}, we conclude that $V(G') = O(m)$ and $E(G') = O(nm)$. 

The problem of finding a max-weight perfect matching and a min-weight perfect matching are symmetric: we can multiply edge weights by $-1$ to switch between the two problems. It follows that we can apply the min-weight perfect matching algorithm proposed by Galil, Micali, and Gabow~\cite{minweightperfectmatching} in Step~\ref{step:matching} of our algorithm. This procedure runs in $O(|V(G')||E(G')| \log|V(G')|) = O(nm^2 \log n)$ time, which dominates the $O(nm)$ construction time of $G'$ in the first step. Hence, the overall running time of Algorithm~\ref{alg:vertex-disjoint-cycles} is $O(nm^2 \log n)$.
\section{Future Work}

Potential avenues for future work include developing tighter upper and lower bounds based on properties of the input graph and devising a more efficient solution to the general problem.

Practical challenges associated with the design of reconfigurable structures (Figure~\ref{fig:deployable-structure}) inspire further intriguing problems. For instance, friction plays a central role in the deployability of such structures --- it determines the force required to draw the string through the system. According to the Capstan equation, friction increases exponentially with the sum of the absolute values of turning angles in the threading route. Therefore, a logical next step is to investigate a variant of \threading\ where the focus is on minimizing this frictional cost instead of the threading length.

\bibliography{bib}

\end{document}